\documentclass[10pt, journal, twocolumn]{IEEEtran}

\usepackage{times}

\usepackage{amsthm}
\usepackage{amsmath}    % Define \boldsymbol (in amsbsy too) and align
\usepackage{amssymb}    % Define \mathbb (in amsfonts too)
\usepackage{mathrsfs}   % Define \mathscr, a script font
\usepackage{epsfig}     % Define \epsfig
\usepackage{graphicx}
\usepackage{algorithm}
\usepackage{algorithmic}
\usepackage{colordvi}
\usepackage{epstopdf}

\allowdisplaybreaks

\IEEEoverridecommandlockouts

\newcommand{\script}[1]{{\mathscr #1}}

\newcommand{\cA}{{\cal A}}
\newcommand{\cB}{{\cal B}}

\newcommand{\cE}{{\cal E}}

\newcommand{\cG}{{\cal G}}

\newcommand{\cN}{{\cal N}}

\newcommand{\cP}{{\cal P}}

\newcommand{\cV}{{\cal V}}

\newcommand{\sC}{\script{C}}

\newcommand{\sG}{\script{G}}

\newcommand{\sN}{\script{N}}

\newcommand{\sP}{\script{P}}

\newcommand{\bX}{\mathbf{X}}
\newcommand{\bY}{\mathbf{Y}}
\newcommand{\bZ}{\mathbf{Z}}

\newcommand{\bx}{\mathbf{x}}
\newcommand{\by}{\mathbf{y}}
\newcommand{\bp}{\mathbf{p}}
\newcommand{\bc}{\mathbf{c}}
\newcommand{\be}{\mathbf{e}}

\renewcommand{\le}{\leqslant}

\renewcommand{\ge}{\geqslant}

\newcommand{\deff}{\mbox{$\stackrel{\rm def}{=}$}}

\newtheorem{thm}{Theorem} % [section]
\newtheorem{cor}{Corollary}
\newtheorem{lem}{Lemma}
\newtheorem{prop}{Proposition}
\newtheorem{rmk}{Remark}
\newtheorem{definition}{Definition}

\newcommand{\eq}[1]{(\ref{#1})}

\newcommand{\Tref}[1]{Theo\-rem\,\ref{#1}}
\newcommand{\Pref}[1]{Pro\-po\-si\-tion\,\ref{#1}}
\newcommand{\Lref}[1]{Lem\-ma\,\ref{#1}}
\newcommand{\Cref}[1]{Co\-ro\-lla\-ry\,\ref{#1}}

\newcommand{\Aref}[1]{Al\-go\-rithm\,\ref{#1}}

\newcommand{\ep}{\varepsilon}

\newcommand{\F}{\mathbb{F}_p}
\newcommand{\R}{\mathbb{R}}

\newcommand{\grad}{\nabla}
\newcommand{\Infeas}{\text{Infeas}}

\begin{document}

%----------------- The Title Declarations ------------------------------

\title
{
Efficient Capacity Computation and Power Optimization for Relay Networks
}

\author{Farzad Parvaresh and Ra\'ul Etkin

\thanks{F. Parvaresh and R. Etkin are with Hewlett-Packard Laboratories, Palo Alto, CA 94304, USA. (emails: \{parvaresh, raul.etkin\}@hp.com). }
\thanks{
Preliminary conference versions of the results in this paper 
appeared in~\cite{PE11}.}
}

\maketitle

%************************************************************************
%                                                                       *
%            Abstract                                                   *
%                                                                       *
%************************************************************************

\begin{abstract}
The capacity or approximations to capacity of various single-source single-destination relay network models has been characterized in terms of the cut-set upper bound. In principle, a direct computation of this bound requires evaluating the cut capacity over exponentially many cuts. We show that the minimum cut capacity of a relay network under some special assumptions can be cast as a minimization of a submodular function, and as a result, can be computed efficiently. We use this result to show that the capacity, or an approximation to the capacity within a constant gap for the Gaussian, wireless erasure, and Avestimehr-Diggavi-Tse deterministic relay network models can be computed in polynomial time. We present some empirical results showing that computing constant-gap approximations to the capacity of Gaussian relay networks with around 300 nodes can be done in order of minutes.

For Gaussian networks, cut-set capacities are also functions of the powers assigned to the nodes. We consider a family of power optimization problems and show that they can be solved in polynomial time. In particular, we show that the minimization of the sum of powers assigned to the nodes subject to a minimum rate constraint (measured in terms of cut-set bounds) can be computed in polynomial time. We propose an heuristic algorithm to solve this problem and measure its performance through simulations on random Gaussian networks. We observe that in the optimal allocations most of the power is assigned to a small subset of relays, which suggests that network simplification may be possible without excessive performance degradation.
\end{abstract}


%************************************************************************
%                                                                       *
%            I. Introduction                                            *
%                                                                       *
%************************************************************************

\section{Introduction}
Relay networks, where one or more source nodes send information to one or more destination nodes with the help of intermediate nodes acting as relays, are often used to model communication in wireless sensor networks. In sensor networks, sensor nodes have limited power sources and often require multi-hop communication with the help of intermediate nodes to reach the data aggregation centers. To guide the design of these networks it is of interest to characterize fundamental communication limits such as the capacity, which represents the maximum reliable communication rate. 

Various communication models for relay networks capture in an abstract setting different aspects of practical systems. The wireless erasure network model of~\cite{Dana} captures the effect of packet losses in the wireless setting. The deterministic network model of Avestimehr, Diggavi and Tse (ADT)~\cite{Salman-deterministic} incorporates broadcast and interference and can be used to gain insights about communication in more complex models that incorporate noise. Among these, of special importance is the Gaussian relay network, which models power limited transmitters and received signals corrupted by additive white Gaussian noise. 

While the capacity of some network models (e.g. wireless erasure and ADT) is well characterized, the capacity of the Gaussian relay network, even in its simplest form with one transmitter, one relay, and one receiver, is in general unknown. The best known capacity upper bound is the so-called {\em cut-set bound}. A cut $\Omega$ of a network can be considered as a subset of nodes which includes the source node and excludes the destination node. For this cut, the capacity 
$F(\Omega)$ is defined as the maximum rate that 
information can be transferred form the nodes in $\Omega$ to the nodes that are not in $\Omega$ conditioned 
on the fact the information on $\Omega^c$ (the nodes that are not in $\Omega$) is known. The cut-set upper bound is the {\em minimum} cut capacity over all the possible cuts.

In the Gaussian setting, there are several capacity lower bounds based on different communication schemes, such as amplify-and-forward, decode-and-forward, compress-and-forward, quantize-and-forward, etc.~\cite{Brett, Cover2, Kramer}. Recently, Avestimehr, et al.~\cite{Salman-Gaussian} made significant progress in the capacity characterization of Gaussian relay networks by showing that a quantization and coding communication scheme can achieve a communication rate within a constant gap of the cut-set upper bound, where the gap only depends on the number of nodes in the network (i.e. it is independent of the channel gains and power levels). However, the evaluation of the achievable communication rate, which is necessary to implement the scheme, requires the computation of the cut-set bound for the network. Assuming that for a given cut the cut capacity is easy to compute, finding the cut-set upper bound can be a challenging problem. For a network with $n$ relays there are $2^n$ different cuts and a greedy algorithm needs exponential time in the number of relays to compute the cut-set bound. 

In this work we show that the achievable rate of the scheme of~\cite{Salman-Gaussian} for the Gaussian relay network can be computed in polynomial time, and as a result, can be computed efficiently. This result is obtained by showing that the cut capacity of a fairly large class of networks under the assumption of independent encoding at the nodes in $\Omega$ is a submodular function. For the special case of layered relay networks, \cite{RV11} showed the equivalent of our submodularity result simultaneously with our conference version of this paper \cite{PE11}. Submodularity properties of conditional entropy (in terms of which cut-capacities are expressed) have also been used in \cite{Salman-Gaussian, ADT11} to bound the cut-capacity of a network in terms of the cut-capacity of the corresponding unfolded graph\footnote{Please, refer to \cite{Salman-Gaussian} for the definition of an unfolded graph.}. 

Existing results on minimization of submodular functions provide algorithms with polynomial running time $O(n^5 \alpha + n^6)$,
 where $\alpha$ is the time that it takes to compute $F(\Omega)$ and $n$ is the number of nodes in the network~\cite{Orlin}. In addition, there exist possibly faster algorithms without polynomial time performance guarantees based on Wolfe's minimization norm algorithm~\cite{Fujishige2}.  In Section~\ref{sec:simulations}, by simulations, we show that the cut-set bound for a Gaussian relay network with around 300 nodes can be computed on a laptop computer in about a minute using a Matlab package for submodular minimization provided in~\cite{Krause}. 
 
%%%%%%%%%%%%%%%%%%%%

Our results, extend and generalize previous results for the ADT model. This model can be seen as a high signal-to-noise-ratio (SNR) approximation of the Gaussian model, incorporating the effects of broadcasting and superposition of signals while de-emphasizing the effects of noise.  Amaudruz et al.~\cite{Fragouli} showed that the cut-set bound for a {\em layered}\footnote{In a layered network, the nodes in one layer are only connected to the nodes in the next adjacent layer. In particular, there is no direct connection from source to destination.} ADT model can be computed efficiently. They have extended graph flow algorithms such as Ford-Fulkerson's in a nontrivial way to find the maximum possible {\em linearly independent} ({LI}) paths in the network. They showed that the capacity of the network is equal to the maximum number of ({LI}) paths and can be computed in time $O(M\cdot |E|\cdot C^5)$, where $M$ is the maximum number of nodes per layer, $|E|$ is the total number of edges and $C$ is the capacity of the network. Moreover, they showed that the capacity can be achieved by using a very simple one-bit processing at the relay nodes. Later Goemans et al.~\cite{Goemans} showed that the deterministic model is a special case of a flow model based on linking systems,
a combinatorial structure with a tight connection to matroids. As a by-product, they obtained the submodularity of the cut capacity for layered ADT networks.  Using this observation they provided various algorithms related to matroid theory to compute the cut capacity of the layered deterministic model based on finding intersection or partition of matroids. These results led to faster algorithms to compute the capacity of large layered ADT networks. In addition, there has been
other extensions on improving the running time of the current algorithms for computing the capacity of ADT networks ~\cite{Ebrahimi, Erez, Shi, Yazdi}.

In addition to showing that the capacity within a constant gap of the Gaussian relay network can be computed in polynomial time, our results allow us to compute in polynomial time the capacity of the wireless erasure network. Furthermore, we provide a simple proof for the computability in polynomial time of the capacity of the layered and non-layered ADT networks. 

Building on the submodularity of the cut-capacity for independent encoding at the nodes, we show that, in the Gaussian setting, it is possible to efficiently optimize the power allocated to the source and relay nodes. We consider two power optimization problems: (i) minimize the total power satisfying a minimum source-destination data rate constraint and power constraints at each node; (ii) maximize the source-destination data rate satisfying total and individual power constraints at the nodes. Since the capacity of the Gaussian relay network is approximately given by the cut-set upper bound with independent encoding at the nodes, we use this cut-set bound to characterize data rate in the optimization problems. We show that these optimization problems can be solved in polynomial time and use simulations to get insights about some properties of the optimal power allocations for networks of various sizes. We observe that optimal power allocations assign most of the power to a small subset of nodes and that setting the power to zero in the remaining nodes (i.e. removing these nodes from the network) often results in a small rate loss. Nazaroglu, et al. showed in \cite{NOEF11} that for the special case of the $N$-relay Gaussian diamond network a fraction $k/(k+1)$ of the total capacity can be approximately achieved by using only $k$ of the total $N$ relays. This suggests that the diamond network can be significantly simplified by tolerating a small performance loss. Our results provide a numerical counterpart to the fundamental performance bounds derived in \cite{NOEF11} and suggest that network simplification may also be possible in more general Gaussian relay networks.

We obtain these results by considering a general framework to compute the cut-set bound. We assign transmit signal random variable $X_i$ to node
 $i \in \cV$ and we assume the probability distribution over the signals $X_1, X_2, \ldots, X_n$ to be independent, i.e
 $p(X_1, X_2, \ldots, X_n) = p_1(X_1) p_2(X_2) \cdots p_n(X_n)$. We also assign received signal random variables
 $Y_i$'s to each node. The network is defined by the transition probability function $f(Y_1, Y_2, \ldots, Y_n | X_1, X_2, \ldots, X_n)$.
 We further assume that the transition probability function is of the form 
 $f_1(Y_1 | X_1, \ldots, X_n)                   
\cdots f_n(Y_n | X_1, \ldots, X_n)$,
 meaning that the received signals are independent conditioned on the transmitted signals in the network. For such networks we
 show that $F(\Omega) = I(\bY_{\Omega^c}; \bX_\Omega | \bX_{\Omega^c})$\footnote{See Section \ref{sec:notation} for a definition of the notation $\bX_\Omega$, $\bY_{\Omega^c}$, etc.} is submodular with respect to $\Omega$.
 Later we show that for ADT networks, the Gaussian relay network and the wireless erasure network, we can find $p_1(X_1)\cdots p_n(X_n)$ such that $\min_\Omega F(\Omega)$ becomes equal to the capacity or the capacity within a constant gap. In other words, the min-cut problem for these networks can be cast as a minimization of a submodular function.
 
The paper is organized as follows. In Section \ref{sec:submodular} we show that for specific type of networks the cut value, $F(\Omega)$, is a submodular function. We then show in Section \ref{sec:models} that for many wireless network models such as the ADT deterministic network,
Gaussian relay network and wireless erasure network the capacity or an approximation to the capacity can be cast as a minimization of $F(\Omega)$. In Section \ref{sec:power_optimization} we study two power optimization problems and show that they can be solved efficiently. Finally, in Section \ref{sec:simulations} we describe results related to solving optimization problems involving submodular functions and perform power optimization in various randomly generated networks of different sizes. We start by introducing the notation used in the rest of the paper.

%************************************************************************
%                                                                       *
%            II. Notations                                              *
%                                                                       *
%************************************************************************

\section{Notation}
\label{sec:notation}
Let $\cV$ denote the set of nodes in the network and $|\cV|$ its cardinality.
 For any subset $A$ of nodes we denote by $\cV\backslash A$ or $A^c$ the set of 
nodes in $\cV$ that are not in $A$. We assume $\cV \backslash \cA\cup\cB = \cV \backslash (\cA\cup\cB)$.
A cut $\Omega$ is defined as a subset of nodes in $\cV$. A cut splits
the nodes in the network into two groups, the nodes that are in $\Omega$ 
and the ones that belong to $\cV\backslash \Omega$. Random variables are shown in
capital letters such as $X_i$ and $Y_i$. We use boldface letter for vectors, e.g. $\mathbf{x}$ is a constant vector and $\mathbf{X}$ is a random vector. We use $\bX_\Omega$ to denote $(X_{v_1}, X_{v_2}, \ldots, X_{v_{|\Omega|}})$ with $v_i\in \Omega$. The function $I(X;Y|Z)$ is the mutual
information between random variables $X$ and $Y$ conditioned on random variable
$Z$. With a slight abuse of notation we use $H(X)$ to denote either the entropy or differential
entropy of the discrete or continuous random variable $X$~\cite{Cover}. By $\F$ we denote a finite field with $p$ elements. Finally, all the $\log(\cdot)$ functions are in base two.

%************************************************************************
%                                                                       *
%            III. Submodularity of cut-set function                      *
%                                                                       *
%************************************************************************

\section{Submodularity of cut-set function }
\label{sec:submodular}

Submodularity arises in many combinatorial optimization problems and large
body of research has been developed on minimizing
or maximizing submodular functions under various constraints.

A submodular function $f \ : \ 2^\cV \rightarrow \mathbb{R}$ is defined as a function over subsets
of $\cV$ with {\em diminishing marginal returns}, i.e. 
if $A, B \subseteq  \cV$ with $A \subseteq  B$ and any 
$v \in \cV \backslash B$,
$$f(A \cup v) - f(A) \ge f(B \cup v) - f(B).$$
 
The theorem below establishes the submodularity of the cut capacity function of a general relay network under some special assumptions. This theorem will be used in Section~\ref{sec:models} to prove that the capacity or an approximation to the capacity of various specific relay network models can be computed by minimizing a submodular function.

\begin{thm}
\label{thm-submodular-mutual-information}
Consider a network consisting of nodes in $\cV$. Each
node sends a message $X_i, i \in \cV$ and receives 
$Y_i, i \in \cV$. If the messages are independent 
$p(X_1, X_2, \ldots, X_{|\cV|}) = p_1(X_1)p_2(X_2) \cdots p_{|\cV|}(X_{|\cV|})$
and conditioned on the sent messages the
received messages are independent, then the function
$$
F(A) \ = \ I(\bX_A; \bY_{\cV\backslash  A} | \bX_{\cV\backslash  A}) \ \ , \ A \subseteq \cV
$$
is submodular.
\end{thm}
 
\begin{proof}
To show that $F(A)$ is submodular we show that $F(A \cup a) -F(A)$ is monotonically
non-increasing in $A$ for $a \notin A$. 
\begin{align*}
F( A\cup & a) =  I(\bX_{A\cup a} ; \bY_{\cV\backslash  A\cup a} | \bX_{\cV\backslash  A\cup a}) 
\\
 \stackrel{(a)}{=} & H(\bX_{A\cup a} | \bX_{\cV\backslash  A\cup a}) - 
              H(\bX_{A\cup a} | \bY_{\cV\backslash  A\cup a} , \bX_{\cV\backslash  A\cup a}) 
\\
  \stackrel{(b)}{=} & H(\bX_{A}) + H(X_a | \bX_{A}) - H(X_a | \bY_{\cV\backslash  A\cup a} , \bX_{\cV\backslash  A\cup a}) \\
    &- H(\bX_A | X_a, \bY_{\cV\backslash  A\cup a} , \bX_{\cV\backslash  A\cup a}) 
\\
  = & H(\bX_{A}) + H(X_a | \bX_{A}) - H(X_a | \bY_{\cV\backslash  A\cup a} , \bX_{\cV\backslash  A\cup a}) \\
   & -  H(\bX_A | \bY_{\cV\backslash  A\cup a} , \bX_{\cV\backslash  A})
\end{align*}
where (a) is the definition of mutual information and (b)
is from the chain rule for the entropy function.
Therefore, % for $\delta_a(A) \ \deff \ F(A\cup a) - F(A)$ we have
\begin{align*}
F(A\cup & a) - F(A)\nonumber\\
 = &
 H(X_a | \bX_A) - H(X_a | \bY_{\cV\backslash  A\cup a}, \bX_{\cV\backslash  A\cup a}) 
  \\
& \hspace{7ex}- H(\bX_A | \bY_{\cV\backslash  A\cup a}, \bX_{\cV\backslash  A}) \\
& \hspace{7ex}+ H(\bX_A | \bY_{\cV\backslash  A \cup a}, Y_a , \bX_{\cV\backslash  A})
\\
  = &  H(X_a | \bX_A) - H(X_a | \bY_{\cV\backslash  A\cup a}, \bX_{\cV\backslash  A\cup a}) \\
&\hspace{7ex}-I(\bX_A; Y_a | \bY_{\cV\backslash A \cup a}, \bX_{\cV\backslash A}) 
\\
  = & H(X_a | \bX_A) - H(X_a | \bY_{\cV\backslash  A\cup a}, \bX_{\cV\backslash  A\cup a}) 
\nonumber \\
&\hspace{7ex}- H(Y_a | \bY_{\cV\backslash  A\cup a}, \bX_{\cV\backslash  A}) \\
&\hspace{7ex}+ H(Y_a| \bX_A, \bY_{\cV\backslash  A\cup a}, \bX_{\cV\backslash  A}) 
\\
  = & 
 \hspace{-3ex}
\underbrace{H(X_a | \bX_A)}_{\mbox{non-increasing in} \ A} 
\hspace{-2ex}-\hspace{1ex}
\underbrace{
H(X_a | \bY_{\cV\backslash  A\cup a}, \bX_{\cV\backslash  A}) 
}_{\mbox{nondecreasing in} \ A} \nonumber\\
&\hspace{7ex}-\underbrace{
H(Y_a | \bY_{\cV\backslash  A\cup a}, \bX_{\cV\backslash  A}) 
}_{\mbox{nondecreasing in} \ A}
+ H(Y_a | \bX_{\cV})
\end{align*}
where the last equality follows because $Y_a$ is independent of $\bY_{\cV\backslash A\cup a}$
conditioned on $\bX_{\cV}$. So, $F(A\cup a) - F(A)$ is non-increasing in $A$
and thus $F(A)$ is submodular.
\end{proof}

In the following example we show that if the signals at the nodes are correlated then 
$F(A)$ is not necessarily a submodular function.

\noindent
{\bf Example.} Consider a symmetric Gaussian diamond network with
two relays such that the channel gains from source to relays are
equal to one and from relays to destination are equal to three.
Letting $\bX_s, \bX_{r_1}$, and $\bX_{r_2}$ be the signals transmitted at the source and relay nodes, 
then the received signals at relays and destination are given by
$$
\left(
\begin{array}{c}
\bY_{r_1} \\ \bY_{r_2} \\ \bY_d
\end{array}
\right) = 
\left(
\begin{array}{ccc}
1 & 0 & 0 \\
1 & 0 & 0 \\
0 & 3 & 3 
\end{array}
\right)
\left(
\begin{array}{c}
\bX_s \\ \bX_{r_1} \\ \bX_{r_2}
\end{array}
\right) +
\left(
\begin{array}{c}
\bZ_{r_1} \\ \bZ_{r_2} \\ \bZ_d
\end{array}
\right)
$$
where $\bZ_{r_1}, \bZ_{r_2}, \bZ_d$ are i.i.d. $\cN(0,1)$.
For this example, we set the probability distribution of $\bX_s, \bX_{r_1}, \bX_{r_2}$
to be jointly Gaussian with zero mean and covariance matrix
$$
\Sigma = \left( 
\begin{array}{ccc}
1 & 0 & 0 \\
0 & 1 & \rho \\
0 & \rho & 1
\end{array}
\right).
$$
Finally, consider the sets $A=\{s,r_1\}$ and $B=\{s,r_2\}$.
Figure~\ref{fig:correlated_sources} shows how the function
$F(A) + F(B) - F(A \cup B) - F(A \cap B)$ varies for
different values of the correlation coefficient (between $\bX_{r_1}$ and $\bX_{r_2}$) $\rho \in [0,1]$.
We see that $F(A)+F(B)$ can be greater than or less than 
$F(A \cup B) + F(A \cap B)$ depending on the value of $\rho$. 
It follows that in general $F(\cdot)$ is not a submodular or a supermodular function 
when the there is correlation among the signals at the nodes.

\begin{figure}[tb]

\centering
\includegraphics[width=0.50\textwidth]{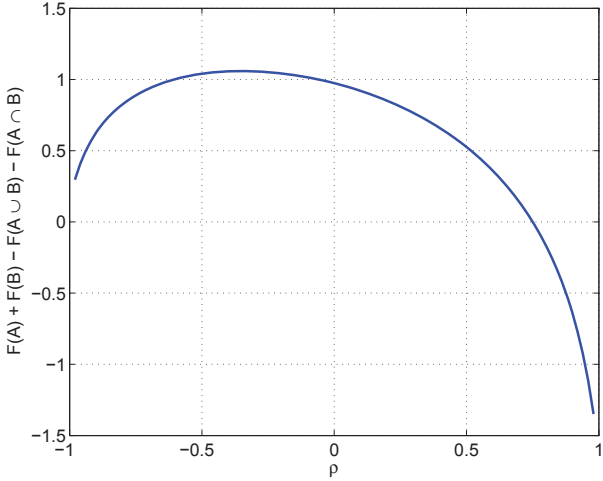}

\caption{$F(A) + F(B) - F(A \cup B) - F(A \cap B)$ as a function of the correlation coefficient $\rho$. In general, $F(\cdot)$ is neither submodular or supermodular when the signals are correlated.}
\label{fig:correlated_sources}

\end{figure}

%************************************************************************
%                                                                       *
%           III. Wireless network models                                *
%                                                                       *
%************************************************************************

\section{Wireless network models}
\label{sec:models}
In this section, by applying the result of \Tref{thm-submodular-mutual-information}, we show that the capacity or an approximation to the capacity for the ADT deterministic network, Gaussian relay network, and wireless erasure network can be cast as a minimization of a submodular function. 

\subsection{Deterministic model (ADT)}
We start by briefly describing the network model of \cite{Salman-deterministic}\footnote{Please, refer \cite{Salman-deterministic} for a more complete description of the model and its motivation.}. In this model, each link from node $i$ to node $j$ has an associated non-negative integer gain $n_{ij}$. Each node $i\in \cV$ transmits a signal $\bX_i$ and receives a signal $\bY_i$, both in $\F^q$ where $q=\max_{i,j} n_{ij}$. At any given time, the received signal at node $j$ is given by
\begin{equation}
\bY_j=\sum_{i\in \cV \backslash\{d\}} \mathbf{S}^{q-n_{ij}}\bX_i
\label{eq:detIO}
\end{equation}
where $d$ is the destination node, the shifting matrix $\mathbf{S}$ is given by
\[
\mathbf{S}=\left(
\begin{array}{ccccc}
0 & 0 & 0 & \cdots & 0\\
1 & 0 & 0 & \cdots & 0\\
0 & 1 & 0 & \cdots & 0\\
\vdots & \ddots & \ddots & \ddots & \vdots\\
0 & \cdots & 0 & 1 & 0\\
\end{array}
\right)
\]
and the sums and products are in $\F$.

For a given cut $\Omega$ of the network, where $\Omega$ includes the source node and excludes the destination node, we can stack together the input vectors $\bX_i, i\in \Omega$ and output vectors $\bY_i, i\in \Omega^c$, and define a transition matrix $\Lambda_\Omega$ that gives the input-output relationship of these vectors according to (\ref{eq:detIO}). It is shown in~\cite{Salman-Gaussian} that the capacity
of the deterministic network is equal to $\min_{\Omega} \text{rank}(\Lambda_\Omega)$. We show next in Theorem \ref{thm:3} that $\text{rank}(\Lambda_\Omega)$ is submodular, and hence the capacity can be computed by minimizing a submodular function. 

\begin{prop}
\label{lemma-deterministic}
Assume an $m\times n$ matrix $A$ over $\F$. Let $\cN$ 
be the subspace
$
\cN \ \deff \ \left\{\bx \in \F^n \ | \ A\bx = 0 \right\}
$,
and let $\sG$ be the set of cosets of  $\cN$ in $\F^n$. Pick $\hat \bx_i$ to be an element in the $i$th coset
of $\cN$ for $i = 1,2, \ldots, |\sG|$, and 
 set $\by_i = A \hat \bx_i$. Notice that $\by_i \neq \by_j$ if $i \neq j$. Now, if we choose $\bx$  uniformly at random from elements of  $\F^n$ with probability $1/|\F^n|$, then the mapping 
$A\bx$ maps $\bx$ to $\{\by_1, \by_2, \ldots, \by_{|\sG|} \}$ uniformly at random with probability $1/|\sG|$. 
In addition, the cosets of $\cN$ form a partition of $\F^n$ into $p^{n}/|\cN|$ sets. Also $\text{rank}(A)+\log_p(|\cN|)=n$. Thus, $\log_p |\sG| = \text{rank}(A)$.
\end{prop}

\begin{thm}
\label{thm:3}
For a deterministic model, given a cut $\Omega$ assume $\Lambda_\Omega$
is the transition matrix form nodes in $\Omega$ to nodes in $\Omega^c$.
Set $D(\Omega) = rank(\Lambda_\Omega)$, then $D(\Omega)$ is submodular.
\end{thm}

\begin{rmk}
A special case of \Tref{thm:3} for {\em layered} ADT networks was 
proved in earlier works~\cite{Goemans,Yazdi}.
\end{rmk}

\begin{proof}
In the network, assume node $i$ sends $b_i$ symbols $x_{i,1}, x_{i,2}, 
\ldots, x_{i,b_i}$ with $x_{i,j}\in \F$. We assume $x_{i,j}$'s drawn i.i.d. with uniform probability distribution over $\F$, i.e. $p(x_{i,j} = q) = 1/|\F|$ for all $q \in \F$. 
From the definition of transition matrix, $\Lambda_\Omega$, if we assume for the cut
$\Omega$, $\mathbf{s} = (s_1, s_2, \ldots, s_k)^t$ symbols are being sent from
nodes in $\Omega$ and
$\mathbf{r} = (r_1, r_2, \ldots, r_\ell)^t$ symbols are being received by nodes
in $\Omega^c$ then $\mathbf{r} = \Lambda_\Omega \mathbf{s}$. Then we can write
\begin{align*}
I(\bX_\Omega; \bY_{\Omega^c} | \bX_{\Omega^c}) = & H(\bY_{\Omega^c} | \bX_{\Omega^c})
- H(\bY_{\Omega^c} | \bX_{\Omega}, \bX_{\Omega^c} ) \\ 
\stackrel{(a)}{=} & H(\bY_{\Omega^c} | \bX_{\Omega^c}) \\
 = & H(\Lambda_\Omega \mathbf{s} | \bX_{\Omega^c} ) \\
\stackrel{(b)}{=} & \log_p |\sG| = \text{rank}(\Lambda_\Omega)
\end{align*}
where $\sG$ is the set of cosets of  $\cN$ where $\cN = \{\mathbf{s} : \Lambda_\Omega \mathbf{s} = 0 \}$.
Equality (a) is because $\bY_{\Omega^c}$ is a deterministic function 
of $\bX_\Omega$ and (b) is the result of \Pref{lemma-deterministic}
and the fact the $\mathbf{s}$ has uniform probability distribution. 
%with $p(\mathbf{s}) = 1/ |\F^k|$.

Notice that for the independent probability distribution on the
sources the received signals are independent conditioned on transmitted
signals so, based on \Tref{thm-submodular-mutual-information}, 
$I(\bX_\Omega; \bY_{\Omega^c} | \bX_{\Omega^c})$
which is equal to $D(\Omega)$ is submodular. 
\end{proof}

\subsection{Gaussian relay network}
\label{subsec:gaussian}
The Gaussian network model captures the effects of broadcasting, superposition and noise of power constrained wireless networks. In this model, at any time index (which we omit) the received signal at node $j \in \cV \backslash \{s\}$ is given by
\begin{equation}
Y_j=\sum_{i\in \cV \backslash\{d\}} h_{ij} X_i+ N_j
\end{equation}
where $X_i\in\mathbb{C}$ is the transmitted signal at node $i$, subject to an average power constraint $E(|X_i|^2)\le 1$, $h_{ij} \in \mathbb{C}$ is the channel gain from node $i$ to node $j$, and $N_j \in \mathcal{CN}(0,1)$ is additive white circularly symmetric complex Gaussian noise, independent for different $j$. 

It has been show in \cite[Theorem 2.1]{OD10} that using lattice codes for transmission and quantization at the relays, all rates $R$ between source $\{s\}$ and destination $\{d\}$ satisfying
\begin{equation}
R \le \min_\Omega I(\bX_\Omega; \bY_{\Omega^c}|\bX_{\Omega^c})- |\cV|
\label{eq:RGaussian}
\end{equation}
can be achieved, where $\Omega$ is a source-destination cut of the network and $\bX_\Omega = \{X_i, i\in \Omega\}$ are i.i.d. $\mathcal{CN}(0, 1)$. 
%It has been shown in \cite{Salman-Gaussian} that 
In addition,
the restriction to i.i.d. Gaussian input distributions is within $|\cV|$ bits/s/Hz of the cut-set upper bound~\cite{Salman-Gaussian}. Therefore the rate achieved using lattice codes in the above result is within $2|\cV|$ bits/s/Hz of the capacity of the network.

The following corollary is an immediate consequence of \Tref{thm-submodular-mutual-information}.
\begin{cor} The function $F(\Omega)=I(\bX_\Omega; \bY_{\Omega^c}|\bX_{\Omega^c})$ with the elements of $\bX_\Omega$ being i.i.d. $\mathcal{CN}(0, 1)$ is submodular.
\label{cor:4}
\end{cor}

Due to \Cref{cor:4} the minimization in (\ref{eq:RGaussian}) is the minimization of a submodular function and the resulting optimal value is within $2 |\cV|$ of the capacity of the network.\footnote{Notice that $I(\bX_\Omega; \bY_{\Omega^c}|\bX_{\Omega^c})=\log \det (I+H H^\dag)$ where $H$ is the matrix of channel gains from nodes in $\Omega$ to nodes in $\Omega^c$ and $H^\dag$ is the conjugate transpose of $H$. Therefore, it is easy to compute the capacity of each cut.}.

\subsection{Wireless erasure network}

In \cite{Dana} the authors introduce 
a special class of wireless networks, called wireless erasure networks. In these networks, a directed graph $\cG=(\cV, \cE)$ defines the interconnections between nodes. To model the broadcast effect of wireless networks, the signals on all outgoing arcs of any given node are equal to each other. There is no interference among multiple arcs arriving at a given node in this model, and the signals on the various arcs are erased independently of each other. We assume binary transmitted signals at each node, i.e. $X_i\in\{0,1\}, i\in\cV\backslash\{d\}$, but all the results can be extended to models with larger input alphabets. It has been shown in \cite{Dana} that the capacity of the network is 
\begin{equation}
\label{eq:erasure_C}
C=\min_\Omega F(\Omega)=\min_\Omega\sum_{i \in \Omega} \left(
1 - \prod_{j \in \Omega^c} \epsilon_{ij} \right)
\end{equation}
where $\epsilon_{ij}$ is the probability of erasure when node $i$ is
sending information to node $j$. We show in the following theorem that $F(\Omega)$ is submodular.
\begin{thm}
The function $F(\Omega)=\sum_{i \in \Omega} \left(
1 - \prod_{j \in \Omega^c} \epsilon_{ij} \right)$ equals $I(\bX_\Omega; \bY_{\Omega^c} | \bX_{\Omega^c})$  where $X_i$ are i.i.d. $\sim \text{Bernoulli}(1/2)$ for $i\in \Omega$. Therefore, $F(\Omega)$ is submodular.
\end{thm}
\begin{proof}
For i.i.d. $X_i\sim \text{Bernoulli}(1/2)$, we can write
\begin{align*}
I(\bX_\Omega;& \bY_{\Omega^c} | \bX_{\Omega^c}) 
\\ = & H(\bX_\Omega | \bX_{\Omega^c}) - H(\bX_\Omega | \bY_{\Omega^c}, \bX_{\Omega^c}) 
\\
\stackrel{(a)}{=} & \sum_{i \in \Omega} \left( H(X_i) - H(X_i | \bY_{\Omega^c}) \right)
\\
\stackrel{(b)}{=} & \sum_{i \in \Omega} \left( 1 - H(X_i | \bY_{\Omega^c}) \right)
\\
= & \sum_{i \in \Omega} \Big( 1 - \\ & \sum_{y_j \in\{1,0,e\}, j \in \Omega^c} 
\hspace{-4ex} H(X_i|Y_j = y_j, j\in \Omega^c)p(Y_j = y_j , j \in \Omega^c) \Big)
\\
= & \sum_{i \in \Omega} \Big( 1 - H(X_i|Y_j = e, j\in \Omega^c)p(Y_j = e , j \in \Omega^c) \Big)
\\
\stackrel{(c)}{=} & 
\sum_{i \in \Omega} \left(
1 - \prod_{j \in \Omega^c} \epsilon_{ij} \right).
\end{align*}
We used in (a) the independence among $X_i$ and the channel erasures, in (b) the fact that for $X_i\sim\text{Bernoulli(1/2)}$, $H(X_i)=1$, and in (c) the fact that for $X_i\sim\text{Bernoulli(1/2)}$,
$H(X_i|Y_j = e, j\in \Omega^c) = 1$ and 
for independent erasures we have $p(Y_j = e , j \in \Omega^c) = \prod_{j\in \Omega^c} \epsilon_{ij}$.
\Tref{thm-submodular-mutual-information}
can be applied to conclude that $F(\Omega)$ is submodular. 
\end{proof}

%************************************************************************
%                                                                       *
%            IV. Power optimization                                     *
%                                                                       *
%************************************************************************

\section{Power optimization}
\label{sec:power_optimization}

In the previous section, for the Gaussian relay network model, we considered {\em fixed power assignments} to the different nodes in the network, and have shown that a constant gap approximation to the capacity can be efficiently computed. In many applications it is of interest to allocate the nodes' transmission powers to optimize a given objective. For example, in a network where the nodes are battery powered, it may be of interest to maximize the network lifetime while satisfying a baseline quality of service. Alternatively, it may be desirable to maximize the network throughput for a given total power budget. This total power budget may arise due to, e.g., a maximum system weight constraint which is dominated by the battery weight, or a total system cost, which may be heavily influenced by the cost of the batteries. Power allocation optimization may also naturally arise in situations where the channel gains, while known to all the nodes, slowly vary over time. In this case, it may be desirable to optimally allocate power for the current channel condition.

% Among the communication models considered in Section \ref{sec:models}, the transmission power only appears explicitly in the Gaussian one. However, power can be related to some parameters in the other models as well. In deterministic ADT networks, the number of quantization levels at each node is a function of the maximum power that the node is using
% to transmit information. In wireless erasure networks, the packet erasure probability, $\epsilon_{ij}$, depends on the received SNR, which in turn depends on the power used by the transmitter node.

As before, we characterize communication rates in terms of cut-set capacities. We consider a model where the cut-set capacities are functions of the cuts and powers assigned to the nodes in the network:
$F(\Omega, \bp) : \cV \times \R^{|\cV|} \to \R$, and we focus on the Gaussian model where this function depends explicitly on the power assignment,
\begin{equation}
\label{eq:gaussian_cut-set}
F(\Omega, \bp) = I(\bX_\Omega; \bY_{\Omega^c} | \bX_{\Omega^c}) = 
\log \det(I + H_\Omega P_\Omega H_{\Omega}^\dagger)
\end{equation}
where $H_{\Omega}$ is the matrix of channel gains from nodes in
$\Omega$ to nodes in $\Omega^c$, $H^\dagger$ is the conjugate transpose
of $H$, and $P_\Omega$ is a diagonal matrix where the diagonal elements are
the powers of the nodes in $\Omega$. 

We will show in \Lref{lem:K_convex} below that in the Gaussian case $F(\Omega, \bp)$ is a concave function of $\bp$. While the results of this section are stated and proved for the Gaussian model, we conjecture that similar results should hold for other models in which $F(\Omega, \bp)$ is a concave function of $\bp$. 

For Gaussian relay networks, we show that the following
optimization problem can be solved in polynomial time,
\begin{align}
\nonumber
& \underset{R,P, \bp}{\text{minimize}} & & \mu_1 R + \mu_2 P \\ \nonumber
& \text{subject to} & & R \le F(\Omega \cup \{s\}, \bp) \ \text{for all}
\ \Omega \subseteq \cV \backslash \{d\} \\ \nonumber
&&& 0 \le \bp \le \bp_{\max} \\ \nonumber
&&& \sum_{i=1}^{|\cV|} p_i \le P \\ 
&&& R_0 \le R, \ P \le P_{tot}
\label{eq:optim}
\end{align}
for fixed constants $\mu_1, \mu_2, R_0$, $\bp_{\max}$ and $P_{tot}$.
In the rest of the section, 
we denote the feasible set of the optimization \eq{eq:optim} by $\mathcal{K}$.

We use the Ellipsoid method \cite{lovasz,ConvOpt} to show that the optimization \eq{eq:optim} can be solved efficiently. 
We will use the following definitions and result. The reader is referred to \cite{ConvOpt} for more details.

\begin{definition}[Polynomial computability] 
\label{def:1}
A family of optimization programs is polynomially computable if:
\begin{itemize} 
\item [(i)] for any instance of the program and any point $\bx$ in the domain, the objective and its subgradient can be computed in polynomial time in the size of the instance. 
\item [(ii)] for a given measure of infeasibility $\Infeas(\cdot)$, it should be possible to determine if $\Infeas(\bx)\le \ep$ in polynomial time, and when this inequality is not satisfied, it should be possible to find in polynomial time a vector $\bc$ such that
\[
\bc^T \bx > \bc^T \by, \forall \by: \Infeas(\by) \le \ep.
\]
\end{itemize}
\end{definition}

\begin{definition}[Polynomial growth]
\label{def:2}
A family of optimization programs has polynomial growth if the objectives and the infeasibility measures as functions of points $\bx$ in the domain grow polynomially with $\|\bx\|_1$.
\end{definition}

\begin{definition}[Polynomial boundedness of feasible sets]
\label{def:3}
A family of optimization programs has polynomially bounded feasible sets if the feasible set of an instance of the program is contained in an Euclidean ball centered at the origin with radius that grows at most polynomially with the size of the instance.
\end{definition}

\begin{prop}[{\cite[Theorem 5.3.1]{ConvOpt}}]
\label{prop:convopt}
Let $\cP$ be a family of convex optimization programs equipped with infeasibility measure $\Infeas(\cdot)$. Assume that the family is polynomially computable with polynomial growth and with polynomially bounded feasible sets. Then $\cP$ is polynomially solvable.
\end{prop}

In order to use \Pref{prop:convopt} we need to check that the optimization \eq{eq:optim} is a convex program. Since the objective function is linear, we only need to check that the feasible set $\mathcal{K}$ is convex.  

\begin{lem}
\label{lem:K_convex}
The feasible set $\mathcal{K}$ is a convex set.
\end{lem}

\begin{proof}
First we show that the function 
$F(\Omega, \bp) = \log \det(I + H_\Omega P_\Omega H_{\Omega}^\dagger)$
is concave in $\bp$ where $0 \le \bp$ for any cut $\Omega \subseteq \cV$.
%We use the fact that $\log\det(X)$ is concave 
%for $X$ belonging to the set of symmetric positive definite matrices~\cite{boyd}.
For any two vectors $\bp_1,\bp_2 \ge 0$ and $\gamma \in [0,1]$ we can write
\begin{align*}
 \gamma F(\Omega, \bp_1) & + (1-\gamma) F(\Omega, \bp_2 ) 
\\
& = \gamma \log\det  ( I + H_\Omega P_1 H_{\Omega}^\dagger ) 
\\
& \quad \quad +(1-\gamma) \log\det(I + H_\Omega P_2 H_{\Omega}^\dagger)  
\\
& \stackrel{(a)}{\le} \log\det(\gamma(I + H_\Omega P_1 H_{\Omega}^\dagger) 
\\
&\quad \quad + (1-\gamma) (I + H_\Omega P_2 H_{\Omega}^\dagger) ) 
\\
& = \log\det(I + H_\Omega(\gamma P_1 + (1-\gamma)P_2) H_\Omega^\dagger)
\end{align*}
where $P_1$ ($P_2$) is a diagonal matrix where the diagonal elements are
the elements of $\bp_1$ ($\bp_2$) that belong to $\Omega$ (respectively),
and (a) follows from the concavity of $\log\det X$, for
$X \succ 0$~\cite{boyd}.

Next,
consider the set $\sC(\Omega) = \{(R, P, \bp) : R \le F(\Omega, \bp) \}$.
Choose two vectors $(R_1, P_1, \bp_1)$
and $(R_2, P_2, \bp_2)$ in $\sC(\Omega)$. We will show that $\sC(\Omega)$ is a convex set by showing that for any $\gamma \in [0,1]$
the vector $\gamma(R_1, P_1, \bp_1) + (1-\gamma)(R_2, P_2, \bp_2)$ is also in $\sC(\Omega)$. Notice that $R_1 \le F(\Omega, \bp_1)$ and $R_2 \le F(\Omega, \bp_2)$ if and only if $(R_1, P_1, \bp_1) \in \sC(\Omega)$
and $(R_2, P_2, \bp_2) \in \sC(\Omega)$. Therefore
\begin{align*}
\gamma R_1 +  (1-\gamma) R_2 
& \le \gamma F(\Omega,\bp_1) + (1-\gamma) F(\Omega, \bp_2)
\\
& \stackrel{(a)}{\le} F(\Omega, \gamma \bp_1 + (1-\gamma ) \bp_2)
\end{align*}
where $(a)$ is due to the fact that $F(\Omega, \bp)$ is a concave
function with respect to $\bp$.  Thus, $\sC(\Omega)$ is a convex set for any $\Omega \subseteq \cV$. 

It is easy to check that the sets 
$\sP_1 = \{(R,P,\bp): 0 \le \bp \le \bp_{\max}\}$, 
$\sP_2 = \{(R, P, \bp): \sum p_i \le P \}$,
$\sP_3 = \{(R,P, \bp): R_0 \le R \}$, and
$\sP_4 = \{(R,P, \bp): P \le P_{tot} \}$ 
are also
convex sets, and, as a result, $\mathcal{K} = \cap_{\Omega \subseteq \cV \backslash \{d\}}
\sC(\Omega \cup \{s\}) \cap \sP_1 \cap \sP_2 \cap \sP_3 \cap \sP_4$ is a convex set.
\end{proof}

Having proved that \eq{eq:optim} is a convex program, in order to use \Pref{prop:convopt} we need to check that the conditions of Definitions \ref{def:1}, \ref{def:2}, and  \ref{def:3} are satisfied. Part (i) of Definition \ref{def:1} follows from the linearity of the objective in \eq{eq:optim}. For part (ii) of Definition \ref{def:1} we specify an infeasibility measure $\Infeas(\cdot):\R^{|\cV|+2} \to \R$ as follows\footnote{With a slight abuse of notation we use $\max$ on vector quantities by taking the maximum over the components of the vector.}:
\begin{align}
\label{eq:infeas}
\Infeas((R, P, \bp))=&\max\bigg\{0, -\bp, \bp-\bp_{\max}, R_0 - R, P-P_{tot}\nonumber\\
&\sum_{i=1}^{|\cV| }p_i -P, R-\min_{\Omega \in \cV\backslash \{d\}} F(\Omega\cup \{s\},  \bp)\bigg\}
\end{align}

The conditions of part (ii) of Definition \ref{def:1} are verified in the following theorem.
\begin{thm}
\label{thm:separation}
For a given vector $(R, P, \bp) \in \R^{|\cV|+2}$ and any $\ep > 0$ we can either  (a) determine in polynomial time
if $\Infeas((R,P,\bp))\le \ep$ and if not (b) 
find in polynomial time a vector $\bc \in \R^{|\cV|+2}$, such that for every $(R', P', \bp')$ satisfying $\Infeas((R', P', \bp'))\le \ep$, $\bc^T (R', P', \bp') < \bc^T (R, P, \bp)$.
\end{thm}

\begin{proof}
Part (a) requires checking that each of the arguments of the $\max$ of (\ref{eq:infeas}) is smaller than or equal to $\ep$ in polynomial time. The first six terms are linear functions and can be easily computed. The last term can be compared to $\ep$ by performing a minimization of a submodular function, which as was shown in Section \ref{subsec:gaussian}, can also be computed in polynomial time.

We focus on condition (b). In this case $\Infeas((R, P, \bp)) > \ep$, meaning that at least one of the arguments of the $\max$ of (\ref{eq:infeas}) is larger than $\ep$. We consider each case separately.

If $-p_i > \ep$ we set $\bc=-\be_{i+2}$ where $\be_i$ has a one in the $i^\text{th}$ position and zeros everywhere else, which can be easily checked to satisfy the condition of part (b).

Similarly, for the cases $p_i - p_{\max,i} > \ep$, $R_0 - R > \ep$, $P-P_{tot} > \ep$ and $\sum_{i=1}^{|\cV|}p_i -P$ we set $\bc=\be_{i+2}$, $\bc=-\be_1$, $\bc=\be_2$, $\bc=(0,-1,1,\ldots,1)$ respectively.

For the last case, let $\Omega^* = \arg \min_{\Omega \in \cV\backslash \{d\}} F(\Omega\cup \{s\},  \bp)$. We have $R-F(\Omega^*\cup \{s\},  \bp)> \ep$. Since the function $F(\Omega^*,\bp)$ is continuous and differentiable with respect
to $\bp$, and the set $\tilde\sC(\Omega^*,\ep)=\{(R, P, \bp) : R \le F(\Omega^*, \bp) +\ep \}$ is convex (which can be shown as in the proof of Lemma \ref{lem:K_convex}),
then the vector ${\bc}=(1, 0, -\grad_{\bp} F(\Omega^*, \bp) )$ is the normal
to a hyperplane that separates $(R,P,\bp)$ from the set $\tilde\sC(\Omega^*,\ep)$. In other words, for all $(R', P', \bp')\in \sC(\Omega^*)$ we have $\bc^T (R', P', \bp') < \bc^T (R, P, \bp)$. 
Noting that $\{(R', P', \bp'):\Infeas((R', P', \bp')) \le \ep\} \subseteq \tilde\sC(\Omega^*,\ep)$, we conclude that part (b) holds in this case as well.
\end{proof}

Having proved these preliminary results, we are ready to prove the main result of this section.

\begin{thm}
\label{thm:polytime}
The optimization in (\ref{eq:optim}) can be solved in polynomial time on the size of the problem.
\end{thm}
\begin{proof}
The proof uses \Pref{prop:convopt}, which requires verifying the convexity of the problem together with the conditions of polynomial computability, polynomial growth, and polynomial boundedness of the feasible set. Convexity was proved in Lemma \ref{lem:K_convex}, while polynomial computability was shown in \Tref{thm:separation}. Polynomial growth follows from the fact that $F(\Omega,\bp)$ is the $\log$ of a polynomial on $\bp$ with degree at most $|\cV|$, while the objective and remaining terms that define the infeasibility measure are linear on $(R, P,\bp)$. Finally, to check that feasible set if polynomially bounded, we note that the feasible set is a subset of the hypercube
\[
\{(R,P,\bp): 0\le (R,P,\bp) \le (R_{\max}, P_{tot}, \bp_{\max}) \}
\]
where $R_{\max}= \min_{\Omega \in \cV\backslash\{d\}} F(\Omega \cup \{s\},\bp_{\max})$. It follows that the feasible set is contained in the Euclidean ball centered at the origin with radius $\|(R_{\max}, P_{tot}, \bp_{\max})\|_2$, which can be easily checked to grow polynomially on the size of the problem.
\end{proof}

The general optimization problem (\ref{eq:optim}) can be specialized to a power minimization with a minimum rate constraint, and to a rate maximization with a total power constraint. Both problems can be solved in polynomial time, as stated in the following corollaries to \Tref{thm:polytime}.

\begin{cor}
The following power minimization problem can be solved in polynomial time. 
\begin{align}
\nonumber
& \underset{\bp}{\text{minimize}} & & \sum_{i=1}^{|\cV|} p_i \\ \nonumber
& \text{subject to} & & R_0 \le F(\Omega \cup \{s\}, \bp) \ \text{for all} \ 
\Omega \subseteq \cV \backslash \{d\} \\
&&& 0 \le \bp \le \bp_{\max} .
\label{eq:power_min}
\end{align}
\end{cor}

\begin{proof}
The corollary follows from \Tref{thm:polytime} by setting $\mu_1=0$, $\mu_2=1$, $P_{tot}=\sum_{i=1}^{|\cV|} p_{\max,i}$.
\end{proof}

\begin{cor}
The following rate maximization can be solved in polynomial time.
\begin{align}
\nonumber
& \underset{R,\bp}{\text{maximize}} & & R \\ \nonumber
& \text{subject to} & 
& 0 \le R \le F(\Omega\ \cup \{s\}, \bp ) \ \text{for all} \ \Omega \subseteq \cV \backslash \{d \} \\ \nonumber
&&& \sum_{i=1}^{|\cV|} p_i \le P_{\text{tot}} \\ 
&&& 0 \le \bp \le \bp_{\max}
\label{eq:rate_max}
\end{align}
\end{cor}

\begin{proof}
The corollary follows from \Tref{thm:polytime} by setting  $\mu_1 = -1, \mu_2 = 0$ and $R_0 = 0$.
\end{proof}

%************************************************************************
%                                                                       *
%            V. Algorithm and simulations                               *
%                                                                       *
%************************************************************************

\section{Algorithms and simulations}
\label{sec:simulations}

In this section we study different algorithms and and provide simulation results
regarding submodular function minimization and power allocation problems.
In the first subsection we look at minimum norm algorithm for submodular function
minimization and we show that this algorithm can find the approximation to capacity
of layered Gaussian relay networks with more that 300 nodes in a couple of minutes.
In the second subsection we propose a heuristic algorithm to find 
the optimum power allocation for Gaussian relay networks. 

%************************************************************************
%                                                                       *
%            V-A. simulation of SFM                                     *
%                                                                       *
%************************************************************************

\subsection{submodular function minimization}

One approach to solve the submodular minimization problem 
due to Lov\'asz is based on {\em extension} of the set function $f:2^\cV\to \mathbb{R}$ 
 to a
convex function $g:[0,1]^{|\cV|}\to \mathbb{R}$
that agrees with $f$ on the vertices of the hypercube $[0,1]^{|\cV|}$, with a
guarantee that $\min_{A \subseteq  \cV} f(A)$ is equal to $\min_\bx g(\bx)$
for $\bx \in [0,1]^{|\cV|}$. In this section we assume the normalization $f(\emptyset)=0$.

The Lov\'asz extension $g$ of {\em any} set function $f$ can be defined as follows.
For a given $\bx\in [0,1]^{|\cV|}$ order the elements of
$\cV$ such that $x(v_1) \ge x(v_2) \ge \cdots \ge x(v_n)$, where $x(v_i)$ is the $v_i$th element of the vector $\bx$. Set $\lambda_0=1-x(v_1)$, $\lambda_i = x(v_i)-x(v_{i+1})$, $\lambda_n=x(v_n)$, and
$$
g(\bx) \ \deff \ \sum_{i=1}^n \lambda_i f(\{v_1, v_2, \ldots, v_i\}).
$$
Define $\mathbf{1}_{\emptyset}=\mathbf{0}\in \mathbb{R}^n$ and $\mathbf{1}_{\{v_1, v_2, \ldots, v_i\}}$ as an $n$ dimensional
vector such that the coordinates $v_1, v_2, \ldots, v_i$ are equal to one and 
all the other coordinates are equal to zero. Then, it is easy to see that $\bx = \sum_{i=0}^n \lambda_i \mathbf{1}_{\{v_1, v_2, \ldots, v_i\}}$ , $\sum_{i=0}^n \lambda_i = 1$ and $\lambda_i \ge 0$. 
 So, $\bx$ is a unique linear convex combination of some vertices of the hypercube and
$g(\bx)$ is linear convex combination of values of $f$ on those vertices.

%A key result is that $f$ is submodular if and only if its Lov\'asz extension $g$ is a convex %function~\cite{lovasz,Fujishige}.  
A key result is that if $f$ is submodular its Lov\'asz extension $g$ is a convex function~\cite{lovasz,Fujishige}.  
In addition, finding the minimum of the submodular function $f$
over subsets of $\cV$ is equivalent to finding the minimum of the convex
function $g$ in the hypercube $[0,1]^{|\cV|}$. 
The optimization can be done in polynomial time using Ellipsoid algorithm~\cite{lovasz}.

There are other algorithms with faster running time to solve the
submodular minimization problem \cite{Iwata,Iwata2,Orlin}. To the best of our knowledge, the running time of the fastest algorithm is in the order of $O(n^5 \alpha +n^6)$, where $\alpha$ is the time that 
the algorithms takes to compute $f(A)$ for any subset $A \subseteq \cV$~\cite{Orlin}. 
For ADT networks, Gaussian relay networks, and erasure networks, $\alpha$ is the time to compute: the rank of $n\times n$
matrices, the determinant of $n \times n$ matrices, and equation
\eq{eq:erasure_C}, respectively. 

However, for networks of large size, a complexity of $O(n^5 \alpha +n^6)$ may still be computationally cumbersome. As a result, in these cases it is desirable to have faster algorithms. Recently, Fujishing~\cite{Fujishige,Fujishige2} showed that the minimization of any submodular function can be cast
as a minimum norm optimization over the base polytope of $f$, $B_f = P_f \cap \{ \bx \ | \ \sum_{i \in \cV} x(i) = f(\cV) \}$, where
$$
P_f \ \deff \ 
\left\{ \bx \in \mathbb{R}^n \ \Big| \ \forall A \subseteq \cV : \sum_{i \in A} x(i) \le f(A) \right\}
$$
and the corresponding minimum norm optimization is 
\begin{eqnarray}
\label{eq:min-norm}
\text{minimize}\ \   ||\bx||_2, \hspace{3ex}
\text{subject to } \bx \in B_f.
\end{eqnarray}
Letting $\bx^*$ be the solution of this minimization, 
the set $A^* = \{v_i  : x^*(v_i) < 0 \}$ is the solution
to $\min_A f(A)$. 
% Although the base polytope $B_f$ has
% exponentially many constraints by using Wolfe's algorithm 
% and the fact that for any vector $x$ the optimization 
% $\max_{y \in P_f} x^T y$, can be done efficiently using the
% Edmonds greedy algorithm ~\cite{Edmonds},
% we can solve the optimization in \eq{eq:min-norm} efficiently. 
Whether the above optimization problem can be solved in polynomial time is an open problem. However empirical studies \cite{Fujishige2} have shown that this algorithm has comparable or even faster running times than the other algorithms with polynomial time performance guarantees. 

\begin{figure}[tb]

\centering
\includegraphics[width=0.45\textwidth]{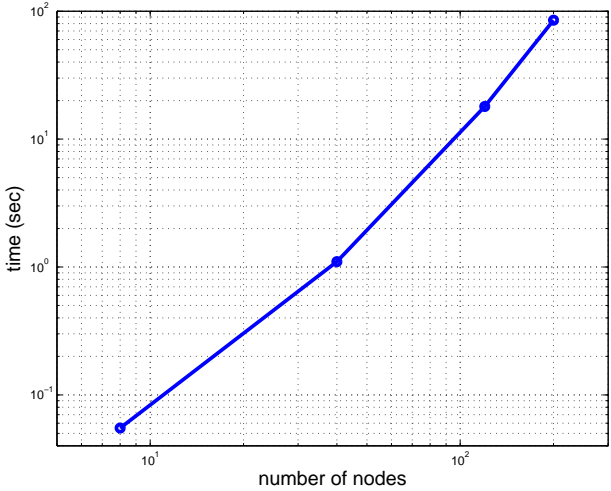}

\caption{Running time of minimum norm algorithm for a layered Gaussian
relay network. Each layer has four nodes.}
\label{fig:1}

\end{figure}

In our specific setting, for layered Gaussian relay networks of size up to around $300$ nodes with 4 nodes per layer, we were able to find the approximate capacity (cf. (\ref{eq:RGaussian})) in order of minutes (see Figure~\ref{fig:1}). 
%on a laptop computer with a 2.8 GHz AMD Dual-Core Processor and 4 GB of memory (
%see Figure~\ref{fig:1}). 
In order to solve the minimization \eq{eq:min-norm} we used 
the Matlab package provided in~\cite{Krause}.

%************************************************************************
%                                                                       *
%            V-B. simulations of power allocation                       *
%                                                                       *
%************************************************************************

\subsection{Power allocation}
In Section~\ref{sec:power_optimization} 
we have shown that the Ellipsoid method can be used
to solve the optimization in \eq{eq:power_min} in polynomial time. While in theory this 
result shows that the optimization in \eq{eq:power_min} is tractable, in practice the Ellipsoid 
method has a number of shortcomings that limit its usability. On the one hand the running time of the Ellipsoid 
method can be large compared to alternative algorithms, and on the other hand, for high dimensional 
problems it has shown numerical instability. In this section, we propose a heuristic 
algorithm to solve the optimization in \eq{eq:power_min} and show that this algorithm converges to the
right solution. While the running time can be exponential on the network size, we show 
through simulations that the algorithm often converges within a time proportional to the network size.

Our proposed algorithm, Algorithm~\ref{alg:SFM_opt} (see pseudo-code below), is based on the cutting plane methods~\cite{boyd} in convex optimization. 
The optimization in \eq{eq:power_min} contains exponentially many constraints of the form
\begin{equation}
\label{eq:constraints}
R_0 \le F(\Omega \cup \{s\}, \bp) \ \text{for} \ \Omega \subseteq \cV \backslash \{d\}.
\end{equation}
In Algorithm~\ref{alg:SFM_opt} we first find the min-cut corresponding to assigning maximum
power to all nodes in the network:
$$
\Omega_1 = \text{argmin}_{\Omega \subseteq \cV \backslash \{d\}} 
F(\Omega \cup \{s\}, \bp_{\max}).
$$
Then, we modify the optimization in \eq{eq:power_min} by replacing the constraint \eq{eq:constraints} with
$$
R_0 \le F(\Omega_1 \cup \{s\}, \bp).
$$
The resulting convex program can be easily solved since it contains few constraints. 

After optimization, we let $\bp^*$ be the optimum power allocation for the current set
of constraints and set
$$
\Omega_i = \text{argmin}_{\Omega \subseteq \cV \backslash \{d\}} 
F(\Omega \cup \{s\}, \bp^*).
$$
We iteratively add the constraint
$$
R_0 \le F(\Omega_i \cup \{s\}, \bp)
$$
to our set of constraints, and solve the optimization again. We stop
if the new constraint is already in the set of constraints.

\begin{algorithm}
\caption{Power minimization}
\label{alg:SFM_opt}
\begin{algorithmic}
\REQUIRE Channel gain matrix $H$, desired rate $R$, vector of nodes' power constraints $\bp_{\max}$.
\ENSURE  Min-cut, Power assignment $\bp^*$ that achieves approximate to rate $R$ with minimum sum of powers.
\STATE $\sC \leftarrow \{\}$, $\bp^* \leftarrow 0$
\STATE $\Omega^* \leftarrow \min_{\Omega \subseteq \cV \backslash\{d\}} F(\Omega\cup\{s\},\bp_{\max})$ 
\IF {$ R \le F(\Omega^*\cup\{s\},\bp_{\max}) $ }
\WHILE{ $\Omega^* \notin \sC$ and $F(\Omega^* \cup \{s\}, \bp^*) < R$}
\STATE $\sC \leftarrow \sC \cup \{\Omega^*\}$
\STATE $\bp^* \leftarrow  \min \sum p_i$
\STATE \ \ \ \ \ \ \ \ subject to:
\STATE \ \ \ \ \ \ \ \ \ \ $R \le F(\Omega \cup \{s\}, \bp)$ for all $\Omega \in \sC$
\STATE \ \ \ \ \ \ \ \ \ \ $0 \le \bp \le \bp_{\max}$
\STATE $\Omega^* \leftarrow \min_{\Omega \subseteq \cV \backslash\{d\}} F(\Omega\cup\{s\},\bp^*)$
\ENDWHILE
\RETURN $\Omega^* \cup \{s\}$, $\bp^*$
\ELSE
\PRINT The constraints are infeasible.
\ENDIF
\end{algorithmic}
\end{algorithm}

Since in each iteration the algorithm adds a new constraint to the constraint set and the number of constraints in \eq{eq:constraints} is finite, the algorithm is guaranteed to find the optimum power in a finite number of iterations, which can be exponential. 

We used simulations to test the performance of the algorithm for networks of varying size $n$ ranging from 10 to 40. For each $n$, we generated 300 random networks with channels gains drawn i.i.d. using a $\cN(0,1)$ distribution. We set the desired transmission rate $R_0=4$ in \eq{eq:power_min} and set a maximum power constraint in each node to $p_{\max}=100$. The results are shown in Figures~\ref{fig:opt_N_const}--\ref{fig:opt_N_node_sparse}, where the the vertical bars represent $\pm 1$ standard deviation around the mean computed over the 300 random networks.

%\ref{fig:opt_sp_rate}.
%\ref{fig:opt_power}, \ref{fig:opt_large_power_one}, \ref{fig:opt_sp_rate},\ref{fig:opt_N_const}.

Figure~\ref{fig:opt_N_const} shows the number of iterations of Algorithm 1 as a function of the number of nodes in the network.  We see that the number of iterations grows $O(n^{3/2})$ with the number of nodes.

\begin{figure}[tb]

\centering
\includegraphics[width=0.47\textwidth]{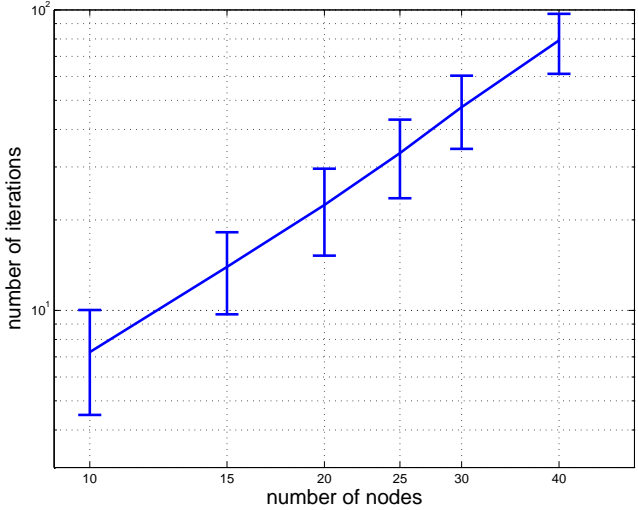}

\caption{Number of constraints in \Aref{alg:SFM_opt} when the algorithm terminates.
The error bars represent one standard deviation.}
\label{fig:opt_N_const}

\end{figure}

\begin{figure}[tb]

\centering
\includegraphics[width=0.47\textwidth]{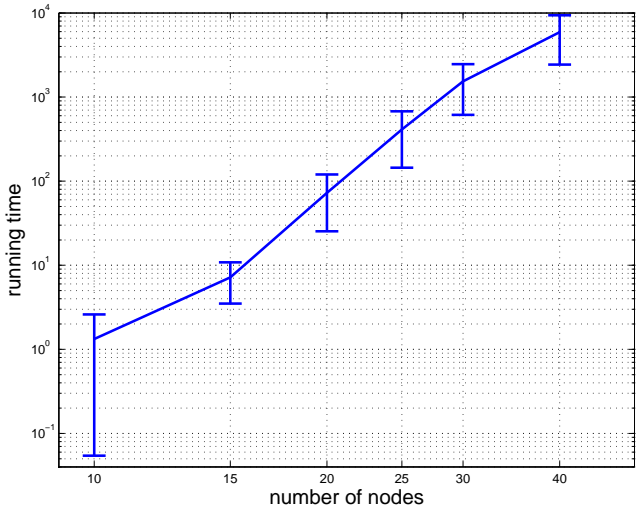}

\caption{Running time in seconds of the power optimization in \Aref{alg:SFM_opt}.}
\label{fig:opt_time}

\end{figure}

\begin{figure}[tb]

\centering
\includegraphics[width=0.47\textwidth]{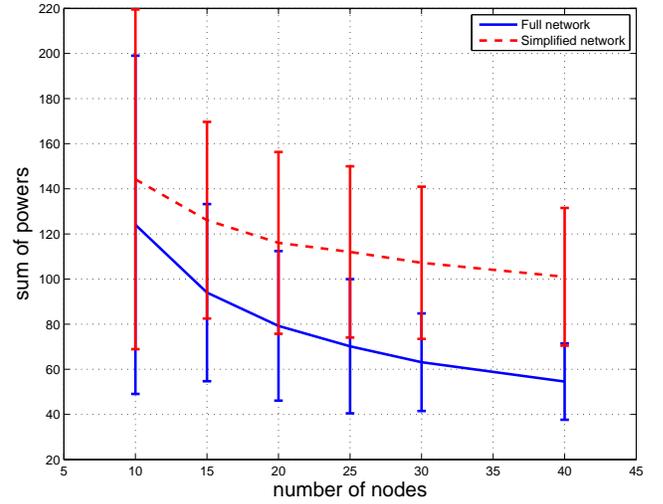}

\caption{Minimum sum of powers for optimization \eq{eq:power_min} when
$R_0 = 4$ and the network is generated randomly as described in the paper.}
\label{fig:opt_power}

\end{figure}

\begin{figure}[tb]

\centering
\includegraphics[width=0.47\textwidth]{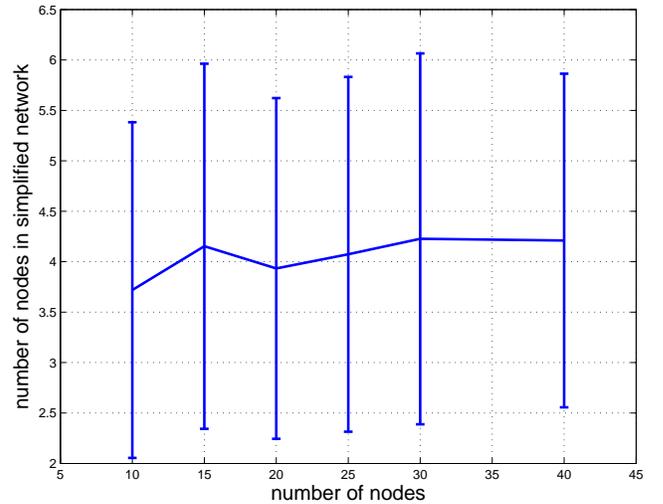}

\caption{Number of nodes in the simplified network.
}
\label{fig:opt_N_node_sparse}

\end{figure}

In Figure~\ref{fig:opt_time}, we present the simulation time as a function of the network size. We observe a running time that grows slower than $O(n^6)$. 
The figure also shows that the power optimization of networks of 40 nodes completes in less than two hours on an Intel Xeon quad core CPU running at 2.33 GHz.

The minimum sum of powers to approximately achieve $R_0=4$ for networks of different size is presented
in Figure~\ref{fig:opt_power} with a blue solid line. Interestingly, the plot shows that the minimum power concentrates around the mean. In addition, the figure shows diminishing returns in total power savings resulting from increasing the network size. 

For the special case of the diamond relay network with $N$ relays, \cite{NOEF11} shows that a fraction $k/(k+1)$ of the network capacity can be approximately achieved by using only $k$ relays. In our setting of a general Gaussian relay network and sum-power minimization, we are interested in investigating whether it is possible to remove a large fraction of the nodes from the network without significantly affecting its performance.

In order to determine which nodes to remove from the network, we solve the sum-power minimization and compare the optimal power allocation $p^*_i$ of each node $i$ to a threshold $P_{th}$. All relays with $p^*_i < P_{th}$ are removed from the network. Let $\sN$ be the set of nodes with $p^*_i \ge P_{th}$.
We optimize the power allocation for the network with node set $\{s\} \cup \sN \cup \{d\}$ and determine whether the problem is feasible for $R_0=4$. If the problem is infeasible, we enlarge $\sN$ by adding more relay nodes in decreasing order of $p^*_i$, until the problem becomes feasible.

Figure~\ref{fig:opt_power} shows in the red dashed curve the resulting minimum sum of powers obtained by setting $P_{th}=1$. Figure~\ref{fig:opt_N_node_sparse} shows the corresponding size of the set $\sN$. We observe in Figure~\ref{fig:opt_N_node_sparse} that the number of nodes with allocated power $p^*_i$ exceeding $P_{th}=1$ (possibly including more relays to make the problem feasible) remains fairly constant as the size of the network $n$ increases. This means that most of the power is allocated to a small subset of the nodes. 

Removing the remaining nodes from the network and optimizing the power allocation again over the resulting simplified network results in the minimum total power plotted in Figure~\ref{fig:opt_power}. This figure shows that even though the number of nodes in the simplified network remains approximately constant as $n$ increases, the total power required to approximately achieve $R_0=4$ decreases with $n$. This is due to the fact that larger $n$ allows to choose the best relays for the simplified network. There is some performance loss in terms of total power due to network simplification but this loss may be compensated by power savings arising from turning off some of the relays. While we have not modeled the power consumption of the clocks, CPU and other subsystems required to keep a relay active, in practice they may become comparable to the power consumed by radio transmissions. This makes network simplification very useful in practice.

% In Figure\ref{fig:opt_large_power_one}, we show number of nodes in the network
% with assigned power more than one after optimization. We observe that
% number of nodes in the network that consume power more than one to
% transmit data at rate of four remain almost the same, around 2 to 4 nodes, independent
% of size of the network. In other words, in large networks only a small number of 
% nodes consume most of the power to transmit data.

% \begin{figure}[t]

% \centering
% \includegraphics[width=0.5\textwidth]{Larger_P_than_one.eps}

% \caption{Number of nodes with powers more than one after optimization in \Aref{alg:SFM_opt}.}
% \label{fig:opt_large_power_one}

% \end{figure}

% \begin{figure}[t]

% \centering
% \includegraphics[width=0.5\textwidth]{opt_large_power.eps}

% \caption{Number of nodes with powers more than two after optimization in \Aref{alg:SFM_opt} and \Aref{alg:SFM_itr}.}
% \label{fig:opt_large_power_two}

% \end{figure}

% After power allocation of algorithm~\ref{alg:SFM_opt},
% we set the power of nodes that consume power less than one to zero and then
% computed the approximate achievable rate in the network. The results are
% plotted in Figure~\ref{fig:opt_sp_rate}. We can see we can use small number of nodes in the network to transit the data and loose a small fraction in the data rate.

% \begin{figure}[t]

% \centering
% \includegraphics[width=0.5\textwidth]{opt_sp_rate_one.eps}

% \caption{Achievable rate if we only use the nodes with powers more than one in the network.}
% \label{fig:opt_sp_rate}

% \end{figure}


\begin{thebibliography}{20}


\bibitem{Fragouli}
A.\, Amaudruz and C.\, Fragouli, 
``Combinatorial algorithms for wireless information flow,''
\emph{SODA '09: Proceedings of the Twentieth Annual ACM-SIAM Symposium on
Discrete Algorithms}, 2009.


\bibitem{Salman-Gaussian} 
A.\, S.\, Avestimehr, S.\, N.\, Diggavi and D.\, N.\, C.\, Tse,
``Approximate Capacity of Gaussian Relay Networks,''
\emph{ISIT '08: IEEE international Symposium on Information Theory}, 
pp. 474--478, July 2008.

\bibitem{ADT11}A.\, S.\, Avestimehr, S.\, N.\, Diggavi and D.\, N.\, C.\, Tse,
``Wireless Network Information Flow: A Deterministic Approach,'' {\em IEEE Trans. Info. Theory}, vol. 57, no. 4, pp. 1872-1905, April 2011.
 
\bibitem{Salman-deterministic}
S.\,Avestimehr, S.N.\,Diggavi and D. N C.\, Tse,
``A deterministic approach to wireless relay networks,''
{\em Forty-Fifth Allerton Conference},
Illinois, September 2007. 

\bibitem{Brett}
Brett E. Schein, 
{\em Distributed coordination in network information theory}, 
{ Ph.D. dissertation}, Massachusetts Institute of Technology,
Cambridge, MA, 2001.


\bibitem{Cover}
T. M.\,Cover and J. A.\,Thomas,
{\em Elements of Information Theory},
New York: Wiley, 1991.

\bibitem{Cover2}
T. M. Cover and A. El Gamal, 
``Capacity Theorems for the Relay Channel,'' 
{\em IEEE Trans. Info. Theory}, vol. 25, no. 5, Sept. 1979,
pp. 572--584.

\bibitem{Dana}
A. F.\, Dana, R. \, Gowaikar, R. \, Palanki, B. \, Hassibi, M. \, Effros, 
``Capacity of wireless erasure networks,'' 
\emph{IEEE Transactions on Information Theory,} 
vol.52, no.3, pp.789--804, March 2006


\bibitem{Ebrahimi}
J.\, Ebrahimi, C.\, Fragouli,
``Combinatorial Algorithms for Wireless Information Flow,'' {\em under submission in ACM Transactions in Algorithms,} 2009. Also available at arXiv:0909.4808v1.

% \bibitem{Edmonds}
% J. Edmonds, 
% ``Submodular functions, matroids, and certain polyhedra,'' 
% {\em Proceedings of the Calgary International Conference on Combinatorial Structures and Their Applications}
%  (R. Guy, H. Hanani, N. Sauer and J. Schoenheim, eds., Gordon and Breach, New York, 1970), pp. 69–87; 
%also in: Combinatorial Optimization—Eureka, You Shrink! (M. Juenger, G. Reinelt, %and G. Rinaldi, eds., 
%{\em Lecture Notes in Computer Science} 2570, Springer, Berlin, 2003), pp. 11–26.

\bibitem{Erez}
E.\, Erez, Y.\, Xu and E. M. \, Yeh,
``Coding for the Deterministic Network Model,''
{\em ITA workshop}, San Diego, 2010.

\bibitem{Fujishige}
S. Fujishige, 
{\em Submodular functions and optimization},(Second Edition), Annals of Discrete Mathematics, Vol. 58, Elsevier, Piscataway, N.J., U.S.A., 2005.

\bibitem{Fujishige2}
S. Fujishige, T. Hayashi, and S. Isotani,
{\em The Minimum-Norm-Point Algorithm Applied to Submodular Function Minimization and Linear Programming,} 
Kyoto University, Kyoto, Japan, 2006.


\bibitem{Goemans}
M.\, X.\, Goemans, S.\, Iwata and R.\, Zenklusen,
``An Algorithmic Framework for Wireless Information Flow,''
\emph{Forty--Seventh Annual Allerton Conference}, October 2009.


\bibitem{lovasz}
 M. Groetschel, L. Lovasz, and A. Schrijver, 
``The ellipsoid method and its consequences in combinatorial optimization,'' 
{\em In Combinatorica}, {\bf 1}:169--197, 1981.
 
\bibitem{ConvOpt} A. Ben-Tal, and A. Nemirovski, ``Lectures on Modern Convex Optimization. Analysis, Algorithms, and Engineering Applications,'' SIAM, Philadelphia, PA, 2001. 
 
\bibitem{Iwata}
S.\, Iwata,
``A faster scaling algorithm for minimizing submodular functions,''
\emph{SIAM Journal on Computing},
32(4):833--840, 2003.

\bibitem{Iwata2}
S.\,Iwata and J.\,B.\, Orlin,
``A simple combinatorial algorithm for submodular function minimization,''
In \emph{SODA '09
%:Proceedings of the Nineteenth Annual 
%ACM-SIAM Symposium on Discrete Algorithms
}, pp 1230--1237, 2009.

\bibitem{Kramer}
G.\, Kramer, M. \, Gastpar, and P. \, Gupta,
``Cooperative strategies and capacity theorems for relay networks,'' 
{\em IEEE Transactions on
Information Theory, } 51(9):3037--3063, September 2005.

\bibitem{Krause}
R.\,A.\, Krause,
Matlab Toolbox for Submodular Function Optimization,
{\em http://www.cs.caltech.edu/~krausea/sfo/}


\bibitem{Orlin}
J.\, Orlin,
``A faster strongly polynomial time algorithm for submodular function minimization,''
\emph{Mathematical Programming}, 118(2): 237--251, 2009.


\bibitem{OD10}
A. \"Ozg\"ur, and S. Diggavi, 
``Approximately achieving Gaussian relay network capacity with lattice codes,'' {\em to be submitted to IEEE Journal on Selected Areas in Communications,} 2010. Available on--line at arXiv:1005.1284v1.


\bibitem{Shi}
C.\, Shi and A.\, Ramamoorthy,
``Improved combinatorial algorithms for wireless information flow,'' 
\emph{Forty--Eighth Annual Allerton Conference}, pp. 875--880, October 2010.

\bibitem{Yazdi}
S.\, M.\, Sadegh Tabatabaei Yazdi and S.\, A.\, Savari,
``A combinatorial study of linear deterministic relay networks,'' {\em IEEE Information Theory Workshop (ITW),} pp. 1--5, Jan 6--8, 2010.

\bibitem{boyd}
S.~Boyd, and L.~Vandenberghe,
\emph{Convex Optimization},
Cambridge university press, New York, 2004

\bibitem{PE11}
F.~Parvaresh and R.~Etkin,
``On computing capacity of relay networks in polynomial time,''
{\em IEEE International Conference on Information Theory (ISIT)},
pp. 1342--1346, Saint Petersburg, Russia, August 2011.


\bibitem{NOEF11} C. Nazaroglu, A. Ozgur, J. Ebrahimi, and C. Fragouli, ``Wireless Network Simplification: the Gaussian N-Relay Diamond Network with Multiple Antennas,'' IEEE International Symposium on Information Theory (ISIT), Saint Petersburg, Russia, pp. 81-85, Aug. 2011.

\bibitem{RV11}A. Raja, and P. Viswanath, ``Compress-and-Forward Scheme for a Relay Network: Approximate Optimality and Connection to Algebraic Flows,'' {\em IEEE International Conference on Information Theory (ISIT)},
pp. 1698--1702, Saint Petersburg, Russia, August 2011.
\end{thebibliography}
\end{document}